\documentclass[conference]{IEEEtran}
\IEEEoverridecommandlockouts

\usepackage{cite}
\usepackage{subfigure}
\usepackage{amsmath,amssymb,amsfonts}
\usepackage{algorithmic}
\usepackage{graphicx}
\usepackage{textcomp}
\usepackage{xcolor}
\usepackage{amsthm}
\theoremstyle{plain}
\newtheorem{theorem}{\textbf{Theorem}}
\newtheorem{lemma}{\textbf{Lemma}}
\newtheorem{definition}{Definition}
\newtheorem{assumption}{Assumption}
\newtheorem{proposition}{\textbf{Proposition}}
\newtheorem{remark}{\textbf{Remark}}
\newtheorem{corollary}{\textbf{Corollary}}
\def\BibTeX{{\rm B\kern-.05em{\sc i\kern-.025em b}\kern-.08em
    T\kern-.1667em\lower.7ex\hbox{E}\kern-.125emX}}
\begin{document}

\title{Joint Source-Channel-Generation Coding: From Distortion-oriented Reconstruction to Semantic-consistent Generation}
\author{Tong Wu, Zhiyong Chen, Guo Lu, Li Song, Feng Yang, Meixia Tao, Wenjun Zhang\\
		Cooperative Medianet Innovation Center, Shanghai Jiao Tong University, Shanghai, China\\
		Email: \{wu\_tong, zhiyongchen, luguo2014, song\_li, yangfeng, mxtao, zhangwenjun\}@sjtu.edu.cn}

\maketitle

\begin{abstract}
Conventional communication systems, including both separation-based coding and AI-driven joint source-channel coding (JSCC), are largely guided by Shannon’s rate-distortion theory. However, relying on generic distortion metrics fails to capture complex human visual perception, often resulting in blurred or unrealistic reconstructions. In this paper, we propose Joint Source-Channel-Generation Coding (JSCGC), a novel paradigm that shifts the focus from deterministic reconstruction to probabilistic generation. JSCGC leverages a generative model at the receiver as a generator rather than a conventional decoder to parameterize the data distribution, enabling direct maximization of mutual information under channel constraints while controlling stochastic sampling to produce outputs residing on the authentic data manifold with high fidelity. We further derive a theoretical lower bound on the maximum semantic inconsistency with given transmitted mutual information, elucidating the fundamental limits of communication in controlling the generative process. Extensive experiments on image transmission demonstrate that JSCGC substantially improves perceptual quality and semantic fidelity, significantly outperforming conventional distortion-oriented JSCC methods.

\end{abstract}

\section{Introduction}
For decades, the architecture of many communication systems has been largely shaped by the reconstruction paradigm, rooted in Shannon's rate-distortion (RD) theory \cite{RD}. This framework treats the recovery of information as a deterministic task: the goal is to find a decoder that produces a point estimate $\hat{x}$ as close as possible to the original $x$ within a given bit rate. Guided by this principle, classical standards and recent deep joint source-channel coding (JSCC) schemes \cite{gundu2019, NTSCC, MambaJSCCWu, video} have focused on minimizing explicit distortion metrics, such as mean squared error (MSE). However, this pursuit of fidelity creates a fundamental conflict with visual perception \cite{blau2019rethinking}.

The inherent limitation of the reconstruction-oriented approach is its tendency to produce average results. When communication resources are constrained, MSE-based optimization often leads to the smoothing trap, where high-frequency details are discarded in favor of a blurred, statistically safe reconstruction. While perceptual metrics like Learned Perceptual Image Patch Similarity (LPIPS) \cite{LPIPSZhang} attempt to bridge this gap, they often introduce unnatural artifacts or high-frequency noise because they operate within the constraints of explicit, often imperfect, distortion functions. Though recent rate-distortion-perception (RDP) based coding schemes \cite{RDPJun1, RDPJun2, RDPC} introduce perceptual constraints, e.g., adversarial losses, into the optimization objective, they still operate within a deterministic reconstruction. This paradigm-level constraint, rather than the specific choice of distortion function, limits their ability to model the structure of the natural  data  manifolds.

\begin{figure}[t]
\centerline{\includegraphics[width=0.45\textwidth]{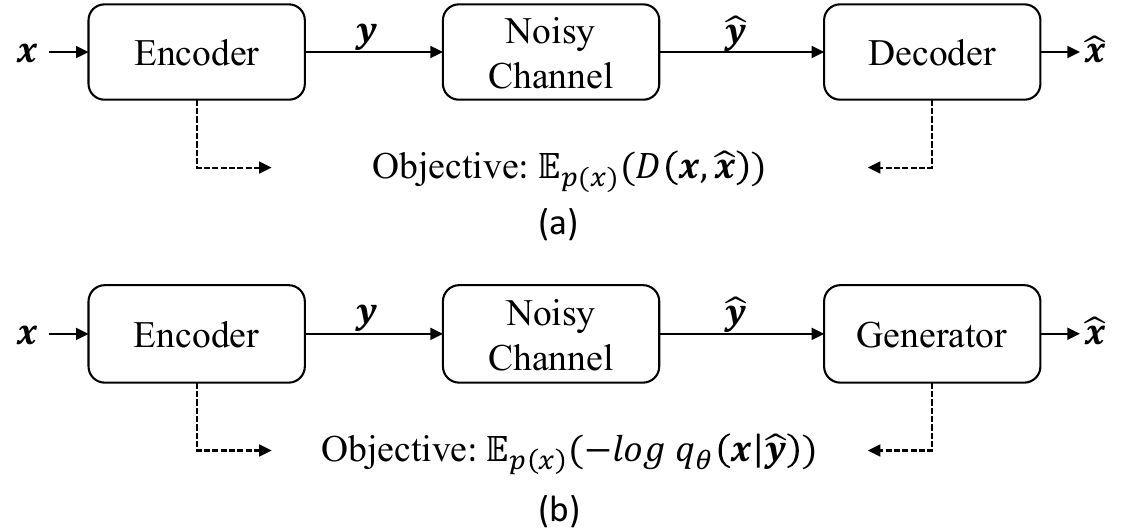}}
\caption{System model: (a) reconstruction paradigm; (b) generation paradigm.}
\label{System_model}
\vspace{-0.5cm}
\end{figure}

In this paper, we propose a new generative paradigm termed \textbf{Joint Source-Channel-Generation Coding (JSCGC)}. Our core philosophy is to advocate for \text{generation instead of reconstruction}. Rather than attempting to recover a single deterministic realization of the source data, JSCGC treats the receiver's task as a conditional stochastic generation problem. In this new paradigm, the received signal is not used to reconstruct data directly, but serves as a semantic guide to sample from the posterior distribution $p(x|\hat{y})$. By shifting from reconstruction to controlled stochastic generation, we enable the system to synthesize high-fidelity textures that are perceptually indistinguishable from authentic data even under severe channel noise, shifting the distortion to inconsistency.

Inspired by these insights, we design a joint coding and generation framework derived directly from the mutual information (MI) function, subject to a perceptual constraint. By leveraging generative models, e.g., diffusion models (DMs) \cite{DDPM, Score-DM, lipman2023flow, CDDM}, to parameterize the data distribution, we ensure that the generated content remains semantically consistent with the source without requiring a rigid loss function. Moreover, this paper provides a theoretical derivation of the lower bound on maximum semantic inconsistency, establishing a mathematical foundation for generative communications. To realize this paradigm, we implement an end-to-end system combining a Mamba-based encoder with a latent flow matching model \cite{z-image}, integrated via a communication-aware adapter that controls the generative process with received signals while effectively preserving large-scale pre-trained knowledge.
\section{System Model}

As shown in Fig.~\ref{System_model}, we consider a communication system in which the transmitter encodes the original data $\mathbf{x}\in \mathbb{R}^m$ into $\mathbf{y}\in\mathbb{C}^l$ via an encoder $E_\phi(\cdot)$ parameterized by $\phi$, where $l<m$. The encoded signal $\mathbf{y}$ is then transmitted through a communication channel, resulting in the received signal $\hat{\mathbf{y}} \in \mathbb{C}^l$. 
At the receiver, unlike traditional JSCC schemes that perform deterministic reconstruction, we formulate the decoding process as a conditional generative problem, in which the recovered data $\hat{\mathbf{x}} \in \mathbb{R}^m$ is sampled from the conditional distribution $q_\theta(\mathbf{x}|\hat{\mathbf{y}})$ using a generative model $G_\theta(\cdot)$ parameterized by $\theta$.

Accordingly, the overall communication process can be modeled as the following Markov chain:
\begin{align}\label{CMC}
\mathbf{X} \rightarrow \mathbf{Y} \rightarrow \hat{\mathbf{Y}} \rightarrow \hat{\mathbf{X}}.
\end{align}
Here, $\mathbf{X},\  \mathbf{Y},\  \hat{\mathbf{Y}},$ and $\hat{\mathbf{X}}$ denote the random variables corresponding to the realizations $\mathbf{x},\  \mathbf{y},\  \hat{\mathbf{y}},$ and $\hat{\mathbf{x}}$, respectively.

As established, the fundamental objective of communication is the reliable transfer of information. Meanwhile, the recovered signal $\hat{\mathbf{X}}$ is intended for human perception, and the system must also ensure that $\hat{\mathbf{X}}$ is perceptually faithful to the source. This necessitates a perceptual constraint, requiring the distribution of the generative signal to closely align with the natural image distribution $p(\mathbf{X})$. 

Therefore, for the communication system defined in (\ref{CMC}), we formulate the joint encoding and generation scheme as the following constrained optimization problem:
\begin{align}\label{training_objective}
  &\max_{\theta, \phi}\  I(\mathbf{X}; \hat{\mathbf{Y}}), \nonumber\\
  &\text{s.t.}\quad d_p\big(p(\mathbf{X}), p(G_\theta(\hat{\mathbf{Y}}))\big) \leq \zeta.
\end{align}
where $d_p(\cdot)$ denotes a divergence metric that quantifies the perceptual discrepancy between the distributions of $\mathbf{X}$ and $\hat{\mathbf{X}}$, and $\zeta$ is a prescribed tolerance.

This design objective sets our work apart from traditional RD or RDP approaches. The primary distinction lies in our departure from explicit distortion functions; instead, we formulate the joint coding and generation scheme by maximizing mutual information subject exclusively to perceptual constraints.
\begin{remark}
Given that we design the coding and generation scheme in JSCGC from the perspective of MI, we primarily aim to explicitly obtain the encoding parameters $\phi$ and generative parameters $\theta$ through optimization from (\ref{training_objective}).
\end{remark}
In practice, the mutual information terms involved in \eqref{training_objective} cannot be computed in closed form due to the unknown distributions of high-dimensional signals. To address this challenge, we propose a practical JSCGC scheme based on deep neural networks to approximate the objective in \eqref{training_objective}.

\section{Proposed JSCGC Scheme}
In this section, we introduce the JSCGC scheme, optimized via joint source-channel-generation design. We provide a perceptually-constrained sampling algorithm for high-quality generation and derive the theoretical lower bound for maximum semantic inconsistency.

\subsection{JSCGC Training Algorithm}
Since the source distribution \(p(\mathbf{X})\) is fixed,  maximizing $I(\mathbf{X}; \hat{\mathbf{Y}})$ is equivalent to minimizing the conditional entropy $H(\mathbf{X} | \hat{\mathbf{Y}})=\mathbb{E}_{(\mathbf{x},\hat{\mathbf{y}})\sim p(\mathbf{x},\hat{\mathbf{y}})}\left[-\log p(\mathbf{x}|\hat{\mathbf{y}})\right]$. However, the true posterior distribution $p(\mathbf{x}|\hat{\mathbf{y}})$ requires knowledge of the ground-truth joint distribution of natural data and channel outputs, which is theoretically intractable.
To address this issue, we have the following proposition.
\begin{proposition}\label{prop:variational_info}
The conditional entropy $H(\mathbf{X}|\hat{\mathbf{Y}})$ can be minimized via the following gradient-based loss:
  \begin{align}
  \mathcal{L}(\theta,\phi)= \mathbb{E}_{t, \mathbf{x}_0, \mathbf{x}_T} \left[ \| (\mathbf{x}_T - \mathbf{x}_0) - \mathbf{\nu}_\theta(\mathbf{x}_t, t, \hat{\mathbf{y}}) \|^2 \right].
  \end{align}
where $\mathbf{x}_0 \sim p(\mathbf{x})$, $\mathbf{x}_T \sim N(0,\mathbf{I})$, $\mathbf{x}_t=(1-\frac{t}{T})\mathbf{x}_0+\frac{t}{T} \mathbf{x}_T$ for $t\in\{1,2,...,T\}$, and $\mathbf{\nu}_\theta(\cdot)$ is the output of the neural network in the generator.
\end{proposition}
\begin{proof}
  We introduce a parameterized variational distribution $q_\theta (\mathbf{x}|\hat{\mathbf{y}})$. Therefore, its evidence lower bound with respect to the variables $\mathbf{x}_t$ is:
  \begin{equation}
  \begin{aligned}
&H(\mathbf{X} | \hat{\mathbf{Y}}) \leq \mathbb{E}\left[-\log q_\theta(\mathbf{x}|\hat{\mathbf{y}})\right] \leq \mathbb{E}_{p}\left[-\log \frac{q_\theta(\mathbf{x}_{0:T}|\hat{\mathbf{y}})}{p(\mathbf{x}_{1:T}|\mathbf{x}_0, \hat{\mathbf{y}})}\right]\\
&=\mathbb{E}_{p} \Bigg[ 
-\underbrace{\log q_\theta(\mathbf{x}_0 | \mathbf{x}_1, \hat{\mathbf{y}})}_{\text{Reconstruction Term}} + \underbrace{D_{\text{KL}}\Big(p(\mathbf{x}_T|\mathbf{x}_0) \,\|\, p(\mathbf{x}_T)\Big)}_{\text{Prior Matching Term}}\\
& + \sum_{t=2}^{T} \underbrace{D_{\text{KL}}\Big(p(\mathbf{x}_{t-1}|\mathbf{x}_t, \mathbf{x}_0, \hat{\mathbf{y}}) \,\|\, q_\theta(\mathbf{x}_{t-1}|\mathbf{x}_t, \hat{\mathbf{y}})\Big)}_{\text{Posterior Term}} \Bigg].
\end{aligned}
\end{equation}
Here, we only need to optimize the posterior term as it quantifies the consistency between the learnable posterior distribution $q_\theta(\cdot)$ and the true posterior $p(\cdot)$, which guarantees the consistency along the trajectory. In this term, the true posterior is tractable. Given $\mathbf{x}_t$ and $\mathbf{x}_0$, $\mathbf{x}_{t-1}$ is deterministic
\begin{equation}\label{post_mean}
\mathbf{x}_{t-1}=\frac{t-1}{t}\mathbf{x}_t+\frac{1}{t}\mathbf{x}_0.
\end{equation}

To prevent mathematical singularities during the KL divergence computation, we formulate the posterior as a Gaussian distribution characterized by a vanishingly small variance $\sigma^2$: 
\begin{equation}
  \begin{aligned}
p(\mathbf{x}_{t-1}|\mathbf{x}_t, \mathbf{x}_0, \hat{\mathbf{y}}) & = \mathcal{N}\left(\mathbf{x}_{t-1};  \mathbf{\mu}(\mathbf{x}_t, t, \mathbf{x}_0), \sigma^2 \mathbf{I}\right)\\
 &= \mathcal{N}\left(\mathbf{x}_{t-1}; \frac{t-1}{t}\mathbf{x}_t+\frac{1}{t}\mathbf{x}_0, \sigma^2 \mathbf{I}\right).
  \end{aligned}
\end{equation}

Therefore, we also formulate the parameterized distribution as a Gaussian distribution with mean $\mathbf{\mu}_\theta(\mathbf{x}_t, t, \hat{\mathbf{y}})$ and variance $\sigma^2 \mathbf{I}$. Therefore, we can compute the posterior term as:
\begin{equation}\label{KL_divergence}
\begin{aligned}
&D_{\text{KL}}\left( \mathcal{N}(\mathbf{\mu}(\mathbf{x}_t, \mathbf{x}_0, t), \sigma^2 \mathbf{I}) \,\|\, \mathcal{N}(\mathbf{\mu}_\theta(\mathbf{x}_t, t, \hat{\mathbf{y}}), \sigma^2 \mathbf{I}) \right) \\
&= \frac{1}{2\sigma^2} \| \mathbf{\mu}(\mathbf{x}_t, \mathbf{x}_0, t) - \mathbf{\mu}_\theta(\mathbf{x}_t, t, \hat{\mathbf{y}}) \|^2. \hspace{-0.5cm}
\end{aligned}
\end{equation}

\addtolength{\topmargin}{0.02in}
To facilitate the training efficiency, we reparameterize the posterior mean to predict the velocity field $\mathbf{v} = \frac{1}{t}(\mathbf{x}_t-\mathbf{x}_0)=\mathbf{x}_T - \mathbf{x}_0$. Therefore, we have 
\begin{equation}
\mathbf{\mu}(\mathbf{x}_t, \mathbf{x}_0, t) = \frac{t-1}{t}\mathbf{x}_t + \frac{1}{t} \left( \mathbf{x}_t - \frac{t}{T}\mathbf{v} \right)= \mathbf{x}_t - \frac{1}{T} \mathbf{v}.
\end{equation}
We further reparameterize the learnable posterior mean as:
\begin{equation}
\mathbf{\mu}_\theta(\mathbf{x}_t, t, \hat{\mathbf{y}}) = \mathbf{x}_t - \frac{1}{T} \mathbf{\nu}_\theta(\mathbf{x}_t, t, \hat{\mathbf{y}}).
\end{equation}
Substituting these into (\ref{KL_divergence}), we can write the tractable optimization objective as:
\begin{equation}\label{training_loss}
\mathcal{L}(\theta,\phi)= \mathbb{E}_{t, \mathbf{x}_0, \mathbf{x}_T} \left[ \| (\mathbf{x}_T - \mathbf{x}_0) - \mathbf{\nu}_\theta(\mathbf{x}_t, t, \hat{\mathbf{y}}) \|^2 \right].
\end{equation}
\end{proof}

The gradient of the loss function with respect to $\phi$ can be computed via the chain rule with channel function $\hat{\mathbf{y}} = W(\mathbf{y})$:
\begin{equation}
  \begin{aligned}
    \nabla_\phi \mathcal{L} = -2 \mathbb{E}_{t, \mathbf{x}_0, \mathbf{x}_T} \Bigg[ \underbrace{\left( (\mathbf{x}_T - \mathbf{x}_0) - \mathbf{\nu}_\theta(\mathbf{x}_t, t, \hat{\mathbf{y}}) \right)}_{\text{Estimation Error}} \\ \cdot \underbrace{\nabla_{\hat{\mathbf{y}}} \mathbf{\nu}_\theta(\mathbf{x}_t, t, \hat{\mathbf{y}})}_{\text{Generation Term}} \cdot \underbrace{\nabla_{\mathbf{y}} W(\mathbf{y})}_{\text{Channel Term}} \cdot \underbrace{\nabla_\phi E_\phi(\mathbf{x}_0)}_{\text{Source Term}} \Bigg].
  \end{aligned}
  \end{equation}

\textbf{This derivation elucidates the joint source-channel-generation nature of our scheme, as the gradient encapsulates the interactions between the source-channel coding and generator simultaneously.}

Unlike the rate-distortion paradigm, which relies on explicit distortion metrics, our approach leverages generative models to approximate the data distribution, enabling end-to-end joint optimization from a mutual information perspective.

\subsection{Sampling Algorithm of the JSCGC}
Now, we need to recover $\hat{\mathbf{X}}$ from $\hat{\mathbf{Y}}$. In this subsection, we formulate the recovery process as stochastic sampling from $p(\mathbf{x}|\hat{\mathbf{y}})$ to not only consider the perceptual constraint in the optimization problem \eqref{training_objective} but also ensure semantic consistency.
 
The remaining challenge is how to generate samples from this specific conditional distribution. Fortunately, our joint training JSCGC scheme naturally addresses this. During training, we not only get a well-trained encoder, but also derive a generator which approximates the true posterior $p(\mathbf{x}|\hat{\mathbf{y}})$ with the parameterized distribution $q_\theta(\mathbf{x}|\hat{\mathbf{y}})$. The generator can estimate the velocity field $\mathbf{\nu}_\theta(\mathbf{x}_t, t, \hat{\mathbf{y}})$ at each time step $t$ according to the training loss function in (\ref{training_loss}).
Therefore, we can derive the ordinary differential equation (ODE) form of the forward diffusion process as:
\begin{equation}
   \frac{d\mathbf{x}_t}{dt} = \frac{\mathbf{v}}{T}.
\end{equation}

We can reverse this ODE from the noise state $\mathbf{x}_T$ to the clean state $\mathbf{x}_0$ as:
\begin{equation}
   \mathbf{x}_0 = \mathbf{x}_T + \frac{1}{T}\int_{T}^{0} \mathbf{v} dt. 
\end{equation}

\addtolength{\topmargin}{0.02in}
By substituting the neural velocity field $\mathbf{\nu}_\theta(\mathbf{x}_t, t, \hat{\mathbf{y}})$ for $\mathbf{v}$ under the guidance of the received signal $\hat{\mathbf{y}}$, we can derive the sampling algorithm as:
\begin{equation}
    \hat{\mathbf{x}} = \mathbf{x}_T + \frac{1}{T}\int_{T}^{0} \mathbf{\nu}_\theta(\mathbf{x}_t, t, \hat{\mathbf{y}}) dt.
\end{equation}
where $\mathbf{x}_T$ is the Gaussian noise sample from the standard Gaussian distribution. The function can be efficiently solved using numerical ODE solvers such as the Euler method with finite iterative steps. 

By sampling $\hat{\mathbf{x}}$ from the learned distribution $q_\theta(\mathbf{x}|\hat{\mathbf{y}})$, which approximates the true posterior $p(\mathbf{x}|\hat{\mathbf{y}})$, the generated content $\hat{\mathbf{x}}$ necessarily satisfies the perceptual constraint in \eqref{training_objective}.

\subsection{Theoretical Bound on Maximum Semantic Inconsistency}
While generator enables perceptually faithful outputs under limited communication resources, it inherently introduces the risk of semantic uncertainty. Quantifying the fundamental limit of this inconsistency is therefore essential for characterizing the performance ceiling of JSCGC. We analyze the theoretical limit of semantic inconsistency for JSCGC  from the manifold assumption \cite{JoshuaB,fefferman2016testing} and the sphere packing argument. 
\begin{assumption}
  Natural data $\mathbf{x}$ is assumed to reside on a low-dimensional manifold $\mathcal{M}$ embedded in the high-dimensional ambient space $\mathbb{R}^m$. The intrinsic dimension of the manifold, denoted by $d$, satisfies $d < m$.
\end{assumption}
Under the manifold assumption, we introduce the sphere packing argument.
\begin{lemma}
\label{lemma:covering_number}
Let the data manifold $\mathcal{M}$ be covered by a set of disjoint semantic balls $\mathcal{B}(\cdot, \epsilon)$ with radius $\epsilon$. According to Kolmogorov's $\epsilon$-entropy theory \cite{KolTik59}, the number of such balls $K(\epsilon)$ is dependent on the radius $\epsilon$ and can be expressed as: 
\begin{equation}
    \lim_{\epsilon \to 0} \frac{\log K(\epsilon)}{-\log \epsilon} = d \implies K(\epsilon) \approx C \cdot \epsilon^{-d},
\end{equation}
where $C$ is a geometric constant dependent on the volume of the manifold $\mathcal{M}$.
\end{lemma}
Accordingly, we define the following event.
\begin{definition}
The probability $P_e$ that the generated result $\hat{\mathbf{x}}$ falls outside the ball $\mathcal{B}(\mathbf{x}, \epsilon)$ containing the original data $\mathbf{x}$ with radius $\epsilon$ is defined as:
\begin{equation}
    P_e = P(\hat{\mathbf{x}} \notin \mathcal{B}(\mathbf{x}, \epsilon)).
\end{equation}
\end{definition} 
Now we consider the semantic inconsistency inherent in the generative process. Due to its stochastic nature, while the generated content may ideally align with the source data, a rigorous performance analysis necessitates focusing on the worst-case scenario. In the absence of communication signals, generated samples would span the entire data manifold. However, the received signal $\hat{\mathbf{y}}$ imposes a semantic constraint that confines the generative process to a localized region. Following the the sphere packing argument in Lemma \ref{lemma:covering_number}, we model this valid region as a ball $\mathcal{B}(\cdot,\delta)$ with radius $\delta$, such that the generated samples reside within $\mathcal{B}(\cdot, \delta)$ with probability $1-P_e$ (where $P_e \to 0$). Given that the maximum deviation occurs at the boundary, we define the maximum semantic inconsistency as the radius $\delta$.

\begin{definition}
  The maximum semantic inconsistency between the generated $\hat{\mathbf{x}}$ and the original data $\mathbf{x}$ is defined as the radius $\delta$ of a ball $\mathcal{B}(\mathbf{x}, \delta)$ such that the generated samples lie within this ball with probability $1-P_e$, where $P_e \to 0$.
\end{definition}
It is evident that a sufficiently large $\delta$ can trivially cover the entire data manifold. Therefore, our objective is to characterize the minimum admissible value of $\delta$.
\begin{definition}
  The lower bound of the maximum semantic inconsistency $\delta^*$ is defined as:
  \begin{equation} 
    \delta^* \triangleq \inf \left\{ r \in \mathbb{R}^+ \;\middle|\; P\left( \hat{\mathbf{x}} \notin \mathcal{B}(\mathbf{x}, r) \right) \to 0 \right\}.
    \end{equation}
\end{definition}

Based on this definition, we can derive a theoretical lower bound on the maximum semantic inconsistency $\delta$.
\begin{theorem}
\label{thm:min_info_bound}
Let $\mathcal{M}$ be a data manifold with intrinsic dimension $d$, and let $C$ denote a geometric constant. Given the mutual information $I(\mathbf{X}; \hat{\mathbf{Y}})$, the lower bound of the maximum semantic inconsistency $\delta^*$ is:
\begin{equation}
  \delta^*=C^{\frac{1}{d}} 2^{-\frac{I(\mathbf{X} ; \hat{\mathbf{Y}})+1}{d}}.
\end{equation}
\end{theorem}
\begin{proof}
 The problem can be interpreted as a classification problem, and Fano's Inequality yields:
  \begin{equation}
    H(\mathbf{X}) - I(\mathbf{X} ; \hat{\mathbf{Y}}) \leq H(P_e) + P_e \log (K(\delta) - 1).
\end{equation}
Rearranging  and relaxing the terms, we derive $P_e$ satisfies:
\begin{equation}
    P_e \geq 1 - \frac{I(\mathbf{X} ; \hat{\mathbf{Y}}) + 1}{\log K(\delta)}.
\end{equation}
To allow $P_e \rightarrow 0$, the number of the ball should satisfy:
\begin{equation}
    K(\delta) \leq 2^{I(\mathbf{X} ; \hat{\mathbf{Y}}) + 1}.
\end{equation}
Consider the relationship between $K$ and $\delta$ in the Lemma \ref{lemma:covering_number}, the radius $\delta$ should satisfy:
\begin{equation}\label{delta_bound}
    \delta \geq C^{\frac{1}{d}} 2^{-\frac{I(\mathbf{X} ; \hat{\mathbf{Y}})+1}{d}}.
\end{equation}
Therefore, the lower bound of $\delta$ is $\delta^*=C^{\frac{1}{d}} 2^{-\frac{I(\mathbf{X} ; \hat{\mathbf{Y}})+1}{d}}$.
\end{proof}
The theorem formalizes a fundamental trade-off where \textbf{increasing mutual information decreases semantic inconsistency, yet the efficiency of this reduction is strictly governed by the manifold's intrinsic dimensionality $d$.}

\begin{corollary}
\label{cor:optimality}
Based on Proposition \ref{prop:variational_info}, the training objective of JSCGC is equivalent to maximizing the mutual information $I(\mathbf{X}; \hat{\mathbf{Y}})$. Coupled with Theorem \ref{thm:min_info_bound}, we demonstrate that JSCGC minimizes the upper bound of the maximum semantic inconsistency under specific channel conditions.
\end{corollary}
\begin{remark}\label{change}
This theorem implies a fundamental shift in error manifestation compared to conventional distortion-based paradigms. In traditional schemes, a reduction in mutual information $I(\mathbf{X}; \hat{\mathbf{Y}})$ typically results in escalating reconstruction distortion. In contrast, because the generative model is inherently constrained to the natural data manifold, the theorem implies that as a decrease in $I(\mathbf{X}; \hat{\mathbf{Y}})$ leads to heightened semantic inconsistency rather than purely increased pixel-wise distortion. Consequently, while the generated content remains perceptually realistic, it gradually drifts semantically from the source signal.
\end{remark}

\section{Implementation and Experimental Results}
\begin{figure*}[t]
  \centering
  \includegraphics[width=0.95\textwidth]{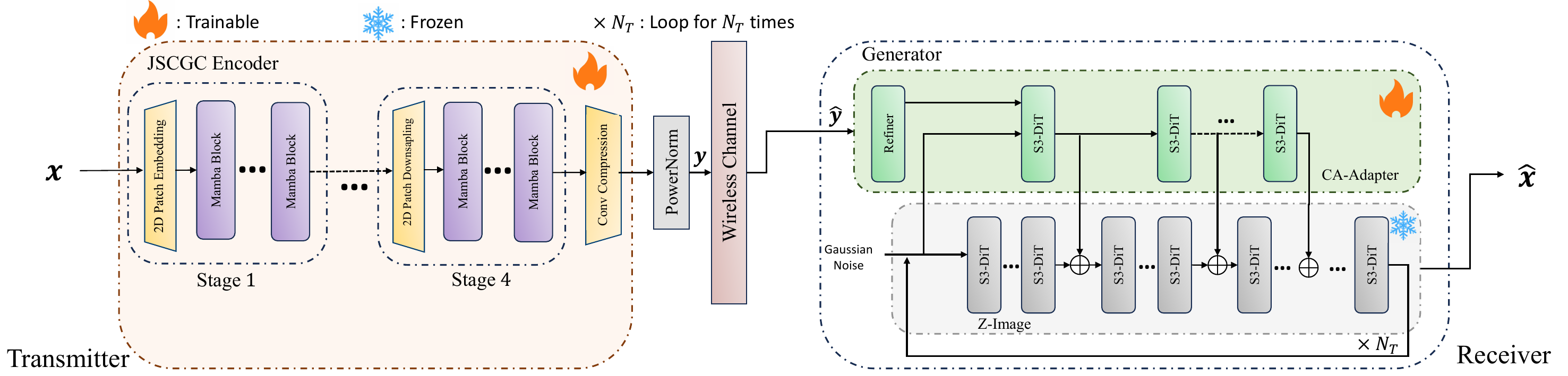}
  \caption{An implementation of the proposed JSCGC scheme.} 
  \label{JSCGC_model}

\end{figure*}

\begin{figure}[t]
  \centering
  \includegraphics[width=0.48\textwidth]{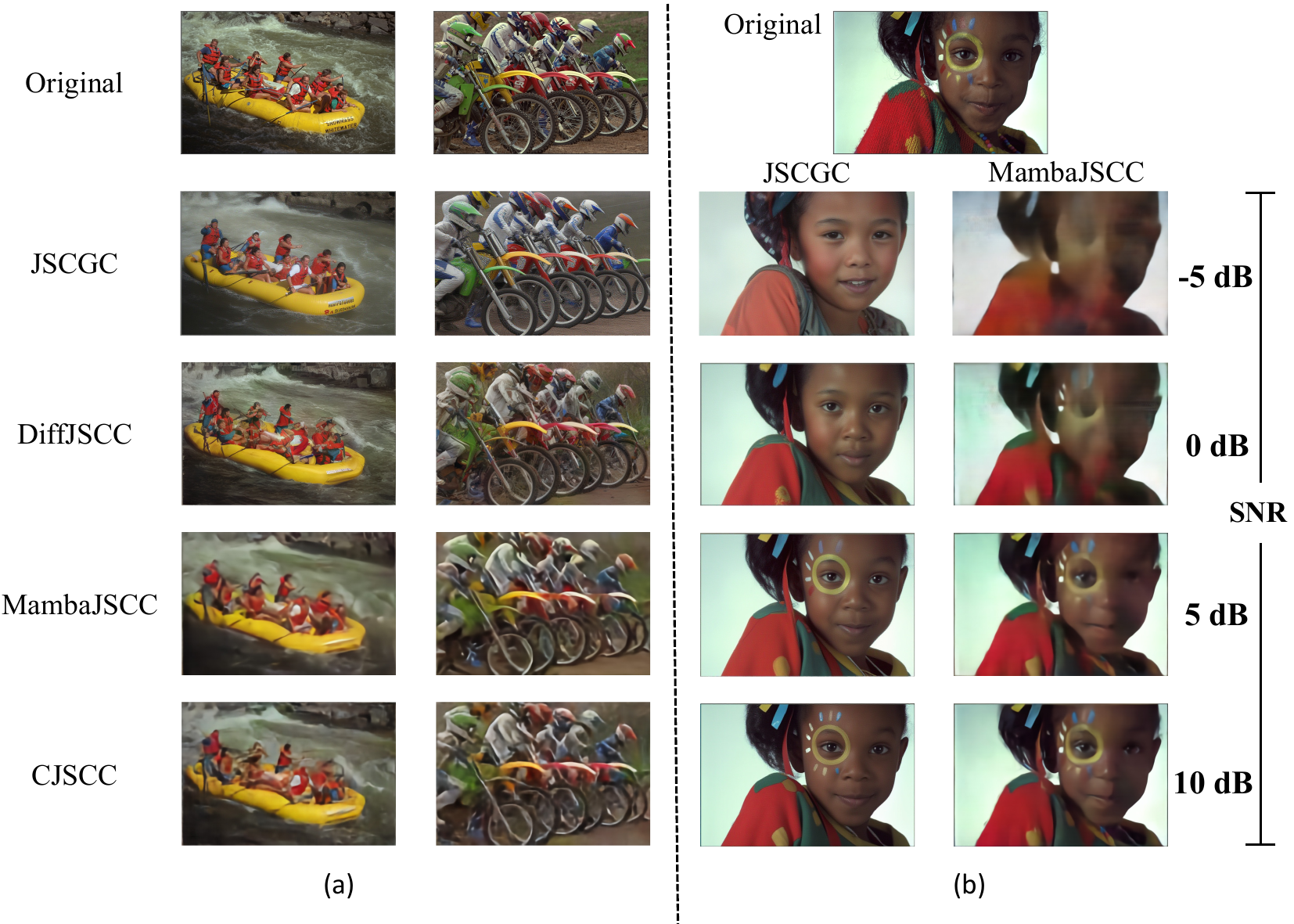}
  \caption{Visualization results under the AWGN channel.} 
  \label{Visualizing}

  \end{figure}

  \begin{figure}[t]
  \centering
  \subfigure[]{\includegraphics[width=0.24\textwidth]{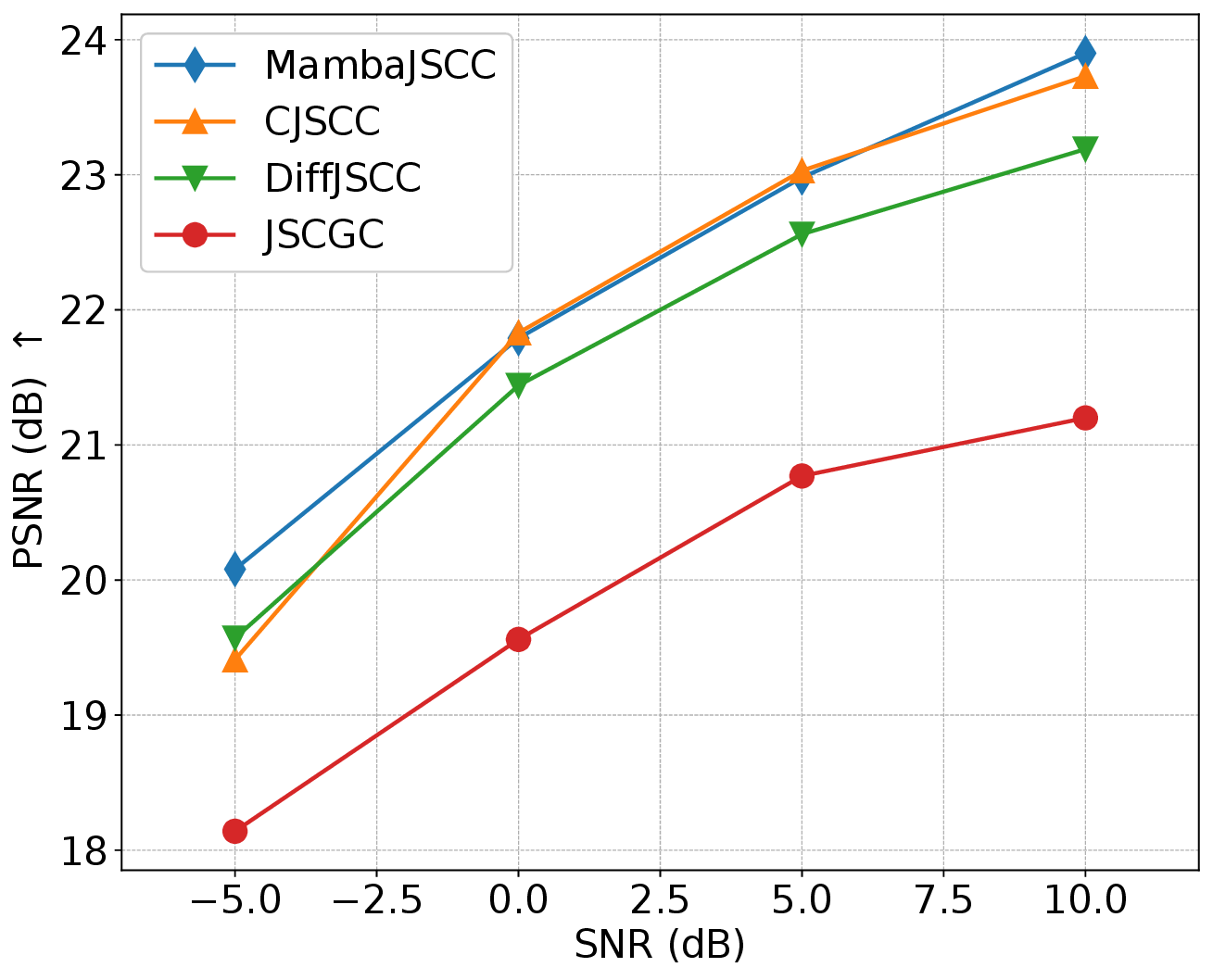}}
  \subfigure[]{\includegraphics[width=0.24\textwidth]{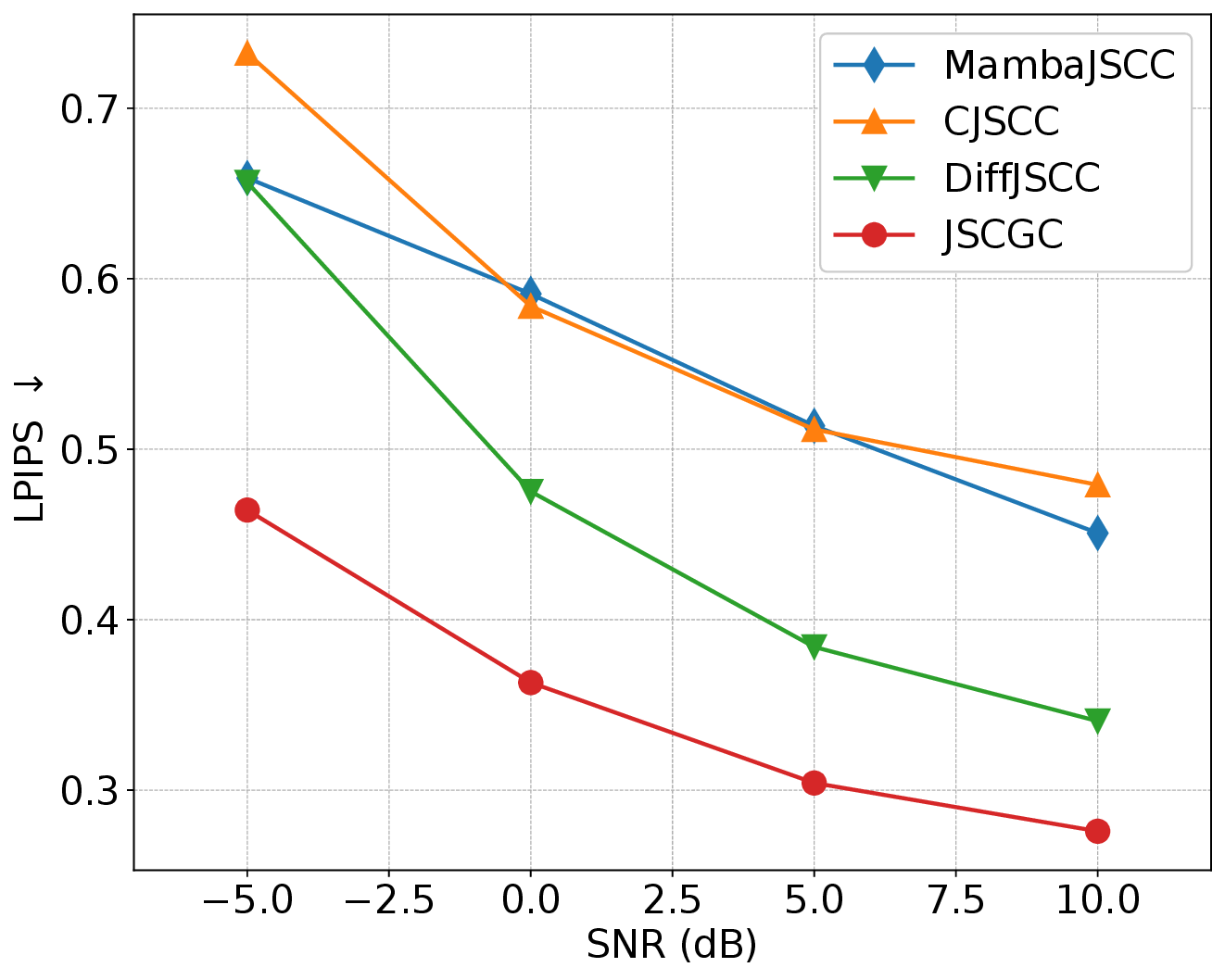}}
  \subfigure[]{\includegraphics[width=0.24\textwidth]{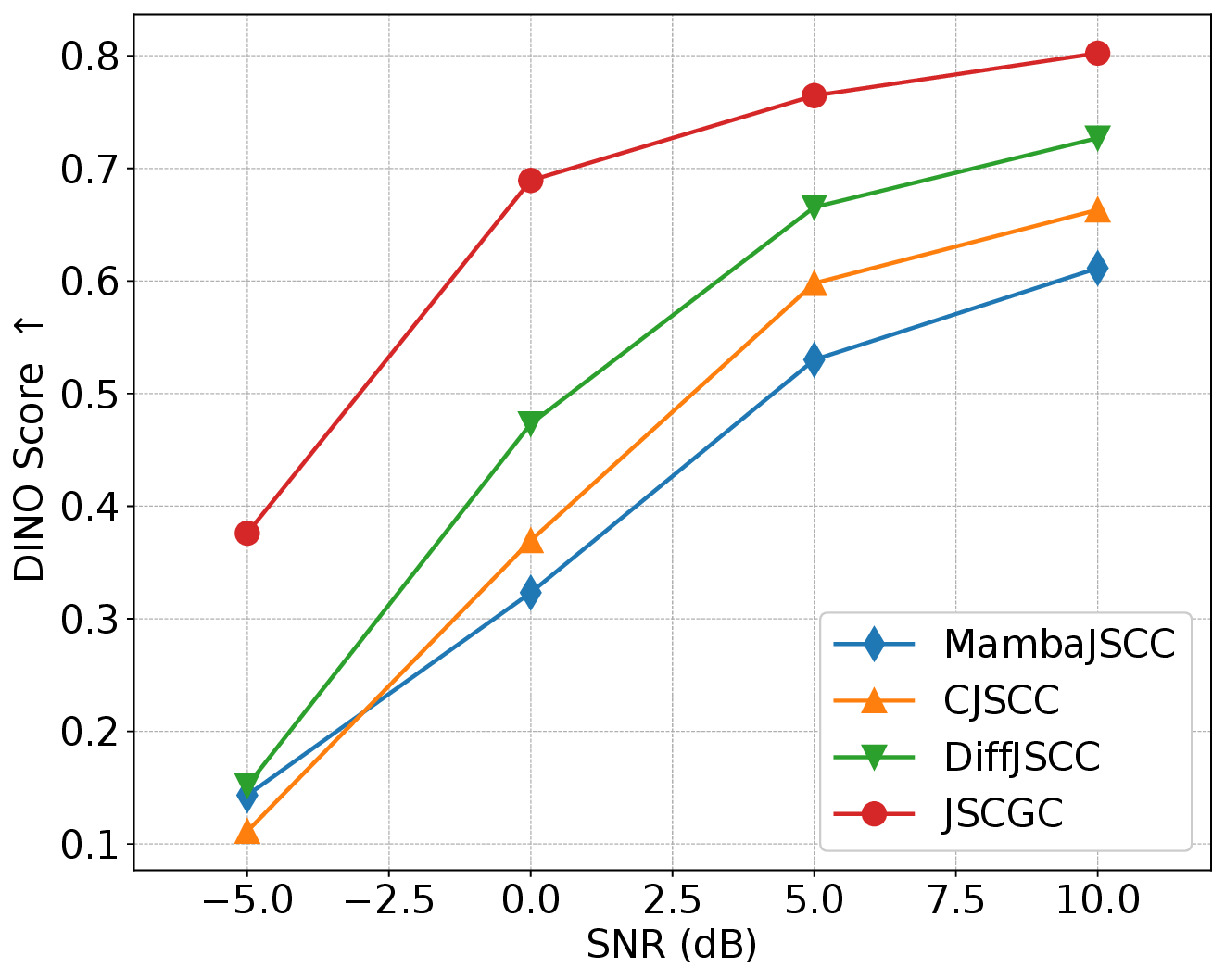}}
  \subfigure[]{\includegraphics[width=0.24\textwidth]{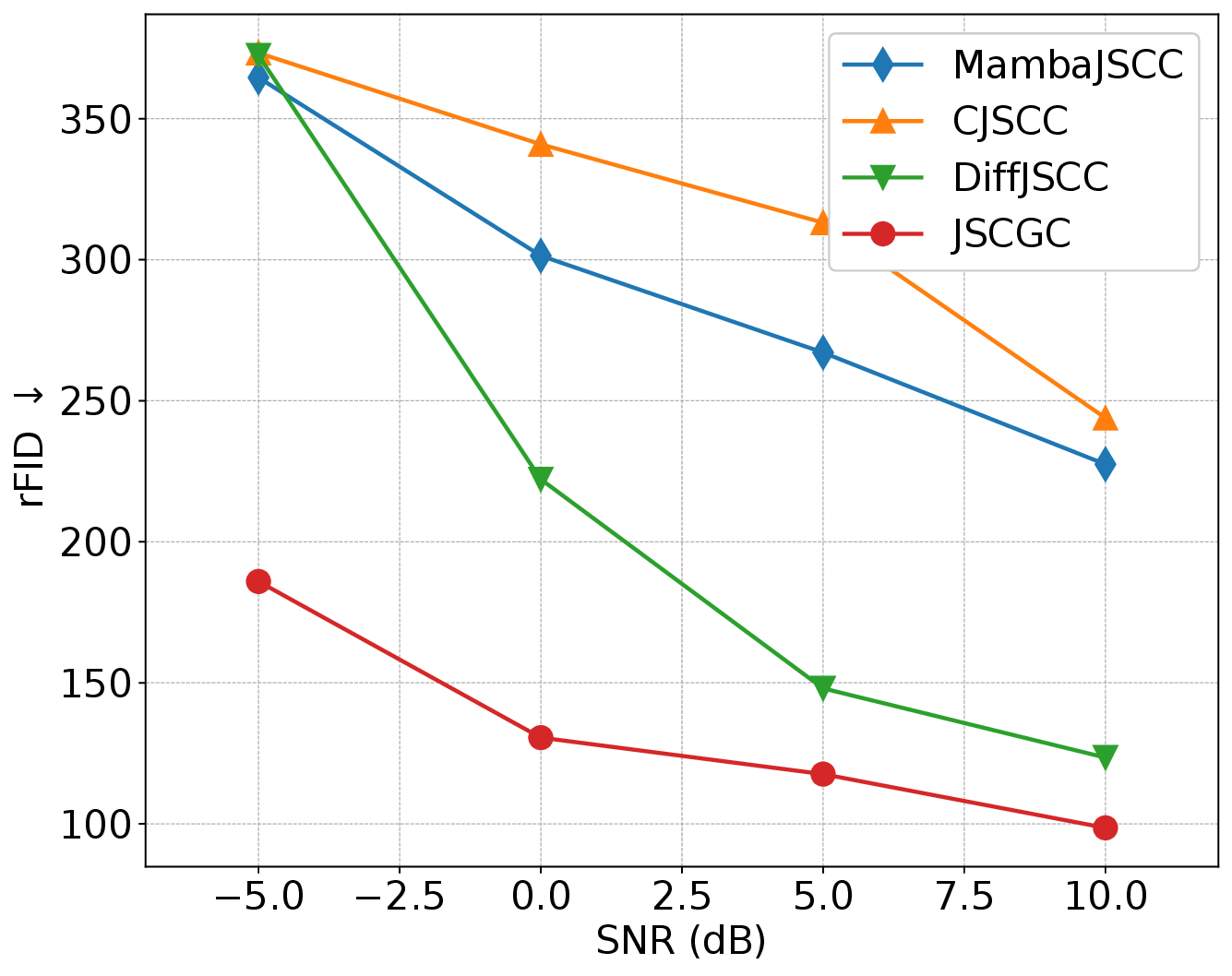}}
  \caption{The performance of different schemes versus SNR under AWGN channel with CBR = $\frac{1}{384}$. (a) PSNR. (b) LPIPS. (c) DINO score. (d) rFID.} 
  \label{awgn_results}

  \end{figure}
\subsection{Architecture Design}

The architectural implementation of the proposed JSCGC framework is illustrated in Fig.~\ref{JSCGC_model}.
To facilitate efficient deployment, we perform image transmission in the latent space, where $\mathbf{x}$ denotes the latent representation of the source image with fixed distribution extracted via a pre-trained extractor. The encoder is designed based on the MambaJSCC architecture~\cite{MambaJSCCWu}, which employs multi-stage Mamba blocks to efficiently extract high-level semantic features. These features are subsequently mapped to the channel input for transmission. At the receiver, the generator is implemented using the Z-Image as the backbone to reduce computational complexity while preserving strong generative capability. Z-Image is a latent-space generative model built upon S3-DiT~\cite{z-image} and pretrained on large-scale datasets, enabling it to capture the distribution of natural data.

However, due to the large-scale nature of the generative model, directly conditioning it on the received signal $\hat{\mathbf{y}}$ is nontrivial. To address this issue while preserving the pretrained knowledge of the generator, we design a communication-aware adapter (CA-Adapter) to modulate the generation process using the received signal. Structurally, the CA-Adapter consists of $N_I$ cascaded stages. At each stage, the adapter takes the received signal $\hat{\mathbf{y}}$ and the intermediate latent variable $\mathbf{x}_t$ as inputs. The resulting feature maps from the $i$-th adapter stage are then injected into the $L_i$-th layer of the Z-Image backbone via element-wise addition, enabling effective conditioning on the communication signal.


\subsection{Experimental Setup}
We utilize a subset of the Open Images dataset consisting of 500k randomly sampled images. All images are randomly cropped to a resolution of $256 \times 256$. The Kodak dataset is used for evaluation. We compare the proposed method with two representative distortion-based JSCC approaches: the CNN-based JSCC (CJSCC) method in \cite{yang2025diffusion} and MambaJSCC~\cite{MambaJSCCWu}, both of which are trained using MSE distortion functions. In addition, we also compare with the generative method DiffJSCC~\cite{yang2025diffusion}. For a fair comparison, the number of generative steps for both Z-Image and the Stable Diffusion model used in DiffJSCC is set to $9$. We consider an additive white Gaussian noise (AWGN) channel with the signal-to-noise ratio (SNR) ranging from $-5$~dB to $10$~dB. The channel bandwidth ratio (CBR) is set to $\frac{n}{l} = \frac{1}{384}$. To comprehensively evaluate the performance of the proposed method, we adopt peak signal-to-noise ratio (PSNR), LPIPS, DINO score \cite{ruiz2023dreambooth}, and reconstruction FID (rFID) \cite{heusel2017gans} as evaluation metrics.

To accelerate convergence, JSCGC is initialized with a pretrained Z-Image model and a pretrained MambaJSCC encoder, and then jointly trained for $10$k iterations for each SNR setting. All experiments are conducted on four NVIDIA RTX 4090 GPUs using DeepSpeed ZeRO-3 with a batch size of 8. 

\subsection{Results}
We first present visualization results under the AWGN channel at $\text{SNR}=5$~dB for different schemes. As shown in Fig.~\ref{Visualizing}(a), distortion-based JSCC methods tend to produce blurry reconstructions. Although DiffJSCC improves image sharpness, its reliance on the blurred outputs of CJSCC results in weaker semantic consistency compared with the ground truth. In contrast, the proposed JSCGC effectively balances visual realism and semantic fidelity, generating more authentic and semantically consistent details. Fig.~\ref{Visualizing}(b) illustrates the performance evolution as the SNR decreases under the AWGN channel. As the SNR decreases, the outputs of distortion-based JSCC methods become increasingly blurred and eventually unrecognizable. In contrast, JSCGC maintains high perceptual quality even at low SNRs; however, the semantic consistency gradually degrades as the guidance signal weakens. This observation is consistent with Remark~\ref{change}.

We then conduct numerical experiments, and the results are shown in Fig.~\ref{awgn_results}, which demonstrate that JSCGC significantly improves perceptual image quality. Although a decrease in PSNR is observed, this behavior is consistent with the well-known rate–distortion–perception trade-off \cite{blau2019rethinking}, as JSCGC is explicitly optimized under a perceptual constraint.

\section{Conclusion}
In this paper, we proposed a novel JSCGC framework that shifts communication from deterministic reconstruction to conditional generation. By maximizing mutual information under channel constraints without relying on explicit distortion metrics, JSCGC enables perceptually faithful communication. We further derived a theoretical lower bound on the maximum semantic inconsistency, revealing the intrinsic resolution limit of generative communication. Experimental results demonstrate that JSCGC achieves superior perceptual quality and exhibits a fundamentally different error behavior from distortion-based JSCC schemes.

\bibliographystyle{IEEEtran}
\bibliography{reference}{}
\end{document}